\documentclass[12pt]{elsarticle}

\usepackage[english]{babel}

\usepackage[letterpaper,top=2cm,bottom=2cm,left=3cm,right=3cm,marginparwidth=1.75cm]{geometry}

\usepackage{amsmath,amsfonts,amssymb,amsthm,euscript,mathrsfs}
\usepackage{graphicx}
\usepackage{algorithm} 
\usepackage{algorithmic}
\usepackage{subfigure}
\usepackage{soul}
\usepackage{graphicx}
    
\newtheorem{theorem}{Theorem}[section]
\newtheorem{lemma}[theorem]{Lemma}
\newtheorem{proposition}[theorem]{Proposition}

\newcommand{\calR}{R}
\begin{document}
%
\begin{frontmatter}

\title{Bounding the number of reticulation events for displaying multiple trees in a phylogenetic network }
\author[label1]{Yufeng Wu}
\address[label1]{School of Computing, University of Connecticut, Storrs, CT 06269 USA}
\ead{yufeng.wu@uconn.edu}

\author[label2]{Louxin Zhang}
\address[label2]{Department of Mathematics and Center for Data Science and Machine Learning,\\ National University of Singapore, Singapore 119076\fnref{label4}}
\ead{matzlx@nus.edu.sg}

\begin{abstract}
 Reconstructing a parsimonious phylogenetic network that displays multiple phylogenetic trees is an important problem in phylogenetics, where the complexity of the inferred networks is  measured by reticulation numbers.
 The reticulation number for a set of trees is defined as the minimum number of reticulations in a phylogenetic network that displays those trees. A mathematical
problem is bounding the reticulation number for multiple  trees over a fixed number of
taxa.  While this problem has been extensively studied for two trees, much less is known about the upper bounds on the reticulation numbers for three or more arbitrary trees. In this paper, we present a few non-trivial upper bounds on reticulation numbers for three or more trees.
\end{abstract}
\end{frontmatter}

\section{Introduction}

Phylogenetic networks are models that extend the classic phylogenetic tree models to accommodate various reticulate evolutionary processes such as horizontal gene transfer, hybrid speciation and recombination \cite{GUSBOOK14,HRS10}. 
Phylogenetic networks are often more biologically realistic than tree models for many applications. However, despite this advantage, phylogenetic networks have not been as widely adopted in biology as phylogenetic trees. One reason for this is that, compared to phylogenetic trees, phylogenetic networks are topologically more complex. In a phylogenetic tree, each node (except for the root, which has an in-degree of zero) has exactly one incoming edge. In contrast, phylogenetic networks contain nodes, known as reticulation nodes, that have two or more incoming edges.

The more reticulation nodes there are in a network, the more complex the underlying evolutionary model becomes. The so-called reticulation number \cite{SempleGS07,HRS10} is introduced as a quantity to measure the complexity of a phylogenetic network. The reticulation number $R(N)$ for a network $N$ is equal to the sum of one less than the in-degree of each node, taken over all nodes in $N$. The smaller $R(N)$ is, the simpler $N$ is. 
Note that $R(N) = 0$ if $N$ is a tree.  Biologists usually prefer simple models, so phylogenetic networks with fewer reticulations are favored. Therefore, in this paper, we focus on the existence of phylogenetic networks with a \textit{small} reticulations number that 'display' a set of phylogenetic trees.

The evolutionary history of a single genomic region is typically modeled by a rooted tree, known as a gene tree. Various methods can be used to infer gene trees from the DNA sequences of multiple species (taxa) within a genomic region.
Now,  suppose we are given a set of gene trees $\mathcal{T} = \{ T_1, \ldots, T_m \}$. A well-known phenomenon in evolutionary biology is that gene trees from different genomic regions tend to be topologically discordant. We are interested in phylogenetic networks that are ``consistent" with these trees.  Alternatively, phylogenetic network can be  viewed as a compact representation of multiple discordant gene trees. A natural approach for reconstructing phylogenetic network for a given set of trees is finding a network $N_{opt}$ that {\it minimizes} $R(N_{opt})$ among all networks $N$ that \emph{display} each of the given gene trees $T_i$. We say $N$ displays a tree $T$ if $T$ can be obtained from $N$ by (i) discarding all but one incoming edges at each node of  $N$, (ii) deleting recursively unlabeled leaves  and (iii) contracting nodes with an in-degree of one and an out-degree of one. We call the reticulation number of $N_{opt}$ the reticulation number, denoted as $R(\mathcal{T})$, for displaying $\mathcal{T}$. See Fig. \ref{fig1} for an illustration.

\begin{figure}[!t]
\centering
\includegraphics[scale=0.6]{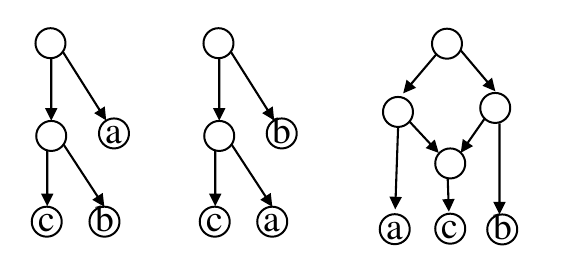}
\caption{ \label{fig1}
Two trees on taxa $a, b, c$ and a phylogenetic network with one reticulation node displaying both trees. Here, we get the left tree by removing the reticulate edge on the left. The  network has one reticulation node.
}
\end{figure}

Computing the exact $R(\mathcal{T})$ for a given $\mathcal{T}$ is known to be challenging. It is NP-hard even when $\mathcal{T}$ contains only two trees \cite{SEMPLENPC07}. There are several methods for computing the reticulation number $R( T_1, T_2 )$ for two trees $T_1$ and $T_2$ that work for trees of moderate sizes (e.g., \cite{ASCH2012,Huson2016,ZEHRET13,WUWANG2010}). For the general case (with two or more trees), there are heuristics for reconstructing networks that work for relatively small trees \cite{ZW2012,WUMIRRN16,WURN10,WURNJCB13}. There are also methods that make simplifying assumptions on network topology. One of the simpler network models is the so-called tree-child network \cite{Cardona_09b}. There are several methods that can construct tree-child networks from the given gene trees (e.g., \cite{Iersel2022,LXZ23}). One main limit of these methods is that there are often little mathematical properties that can be said about the reconstructed networks. There are also methods that can compute the \emph{bounds} of $R(\mathcal{T})$ \cite{WURN10,TCLBJCB23}. While these bounds 
provide the ranges of values of $R(\mathcal{T})$, the bounds are computational rather than mathematical: they are computed for a specific $\mathcal{T}$ and do not offer a lower/upper bound on the reticulation number for an arbitrary set of trees.

An interesting mathematical question is: How small can $R(\mathcal{T})$ be for displaying a set $\mathcal{T}$ of $m$ arbitrary trees on $n$ taxa? 
This problem seems to be intuitive but is very challenging. When $\mathcal{T} = \{T_1, T_2\}$, it is well known that $R( T_1, T_2  ) \le n-2$ where $n$ is the number of taxa \cite{BGMS05}. It was shown in \cite{BGMS05} that the $n-2$ bound is tight because two  ``conjugate" caterpillar trees on $n$ taxa require $n-2$ reticulations to be displayed in any phylogenetic network. See Section \ref{sectPrelimi} for the definition of conjugate caterpillar trees. Surprisingly, little has been known for the more general case of three or more trees ($m = |\mathcal{T}| \ge 3$).
Only a trivial upper bound $(m-1)(n-2)$ on $R(\mathcal{T})$ is known \cite{THREETREES23,TCLBJCB23}.   For the special case of \emph{three} trees, there is a recent paper \cite{THREETREES23} that establishes a \emph{lower bound} on the upper bound of $R(\mathcal{T})$, which is asymptotically close to $\frac{3n}{2}$. More 
specifically, there exists three trees with $n$ taxa such that any network displaying them contains at least roughly $\frac{3n}{2}$ reticulations \cite{THREETREES23}. However, it is open  whether the trivial upper bound of  $2(n-2)$ for the case of three trees is tight or not \cite{THREETREES23}. 

In this paper, we present several upper bounds on $R(\mathcal{T})$ for displaying multiple trees $\mathcal{T}$. More specifically, our contributions include:

\begin{enumerate}
\item A simple  relation between the structure of multiple trees and the reticulation number of these trees, and a simple observation on displaying two trees.
\item An upper bound $2n-2-\Theta(\log\log(n))$ for three trees that reduces the trivial $2(n-2)$ bound  slightly by $\Theta(\log\log(n))$. 
\item An upper bound for any $m$ trees that 
is significantly better than the trivial bound $(m-1)(n-2)$ when $m$ is large (Theorem~\ref{theoremLargeNumTrees}). 
\end{enumerate}

To the best of our knowledge, these are the first non-trivial bounds on the reticulation number for three or more trees. Although the gap between our bounds and the trivial bounds is not large, our results demonstrate the possibility of deriving bounds on the reticulation number that improve upon the trivial bounds for displaying multiple trees.


\section{Preliminaries}
\label{sectPrelimi}

A gene tree is assumed to be a rooted and binary phylogenetic tree,  with leaves uniquely labeled by taxa and edges oriented away from the root.   In this paper, we will consider the phylogenetic networks that display a set of gene trees on the same set of taxa.

A phylogenetic tree is called a caterpillar tree if every internal node has at least one leaf as its child. Caterpillar trees are also known as line tree.  In a caterpillar tree $C$ on a set $\mathcal{X}$ of $n$ taxa, all $n-1$ non-leaf nodes (called internal nodes) are at different depths and are connected into a directed path from the root $r$ to the deepest internal node:
$$r=v_1 \rightarrow v_2 \rightarrow \cdots \rightarrow v_{n-1}.$$
Let $\ell_i$ be the leaf child of $v_i$ for $i\leq n-2$ and the children of $v_{n-1}$ be $\ell_{n-1}$ and $\ell_n$. 
Then, $C$ is uniquely determined by 
the permutation string $\ell_1\ell_2\cdots \ell_{n-2}$ if $\mathcal{X}$ is given, as  $\{\ell_{n-1}, \ell_{n}\}=\mathcal{X}\setminus \{\ell_1, \ell_2, \cdots, \ell_{n-2}\}$.
Consequently,  we use $C(\ell_1\ell_2\ell_3\cdots \ell_{n-2}\{\ell_{n-1}\ell_n\})$ to denote the caterpillar tree $C$. The caterpillar trees $C(\ell_1\ell_2\ell_3\cdots \ell_{n-2}\{\ell_{n-1}\ell_n\})$ and 
$C(\ell_{n}\ell_{n-1}\ell_{n-2}\cdots \ell_{3}\{\ell_{2}\ell_1\})$ are said to be conjugate. For example, the two caterpillar trees in Fig.~\ref{fig1} are conjugate. The two caterpillar trees in Fig.~\ref{fig2} are also conjugate. 

A phylogenetic network $N$ is a leaf-labeled rooted directed acyclic graph. The leaves are the nodes with an out-degree of zero.   
The tree nodes are the nodes with an in-degree of 1 and an out-degree of 2.
Reticulation nodes are those with an in-degree of 2 or more and an out-degree of 1. An edge is called a reticulate edge if it enters a reticulation node. Finally, the network $N$ contains a single node, known as the root, which has an in-degree of 0 and an out-degree 2 or more.


A phylogenetic network is tree-child
if every non-leaf node has at least one tree node or leaf as its child. For example, the two networks in Fig. \ref{fig2} are both tree-child networks. 

Let $N$ be a phylogenetic network. We use $V(N)$ and $E(N)$ to denote the set of vertices and edges of $N$. 
The reticulation number $R(N)$ is defined as: 
\begin{eqnarray}
R(N)=\sum_{v\in V(N)} (i(v)-1). \label{defRN}
\end{eqnarray}
where $i(v)$ denotes the in-degree of $v$.

\begin{figure}[!b]
\centering
\includegraphics[scale=0.6]{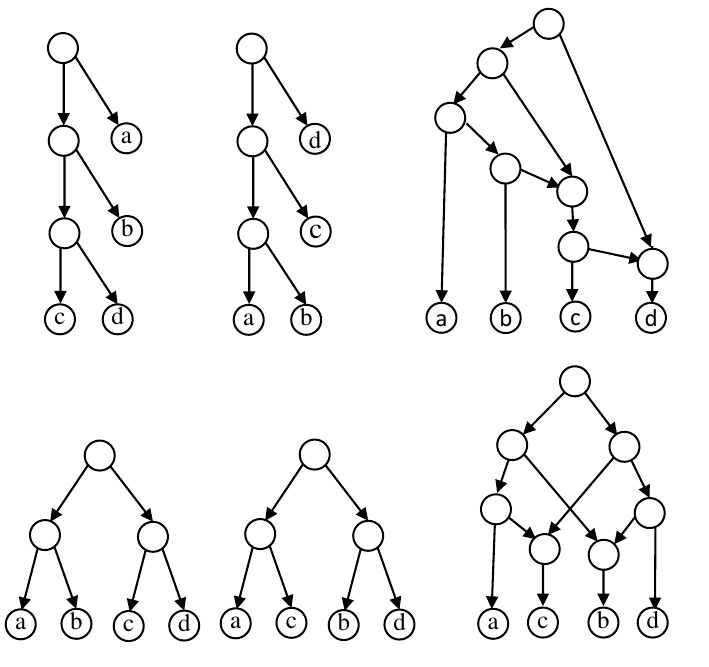}
\caption{\label{fig2}
Top: Two conjugate caterpillar trees on four taxa $\{a, b, c, d\}$ and a phylogenetic network with 2 reticulation nodes displaying them.
Bottom:  two balanced trees on the 4 taxa and a phylogenetic networks with 2 reticulation nodes displaying them.}
\end{figure}

\subsection{The two-tree case}

The reticulation number for two trees has been intensively studied. In this case,  a fundamental connection between the reticulation number and the so-called maximum acyclic agreement forest (MAAF) has been the foundation of most papers for the two-tree case (e.g., \cite{BGMS05,ZW2012,ZEHRET13}). Since MAAF is widely used, we provide only a brief introduction here. For more details, see \cite{BGMS05}.

An agreement forest of a set of rooted trees $\mathcal{T}$ is a collection of rooted binary phylogenetic trees $\mathcal{F}( \mathcal{T} ) = \{F_1, \ldots, F_k\}$ such that (i) the taxa of all $F_i$ in $\mathcal{F}( \mathcal{T} )$ form a disjoint \emph{partition} of the taxa of $\mathcal{T}$,  and (ii) the forest $\mathcal{F}( \mathcal{T} )$ can be obtained  from \emph{each} $T \in \mathcal{T}$ by deleting some branches of $T$, discarding  recursively unlabeled leaves,  and then contracting all degree-two nodes.  

The trees in an agreement forest $\mathcal{F}(\mathcal{T} )$ are called the components of $\mathcal{F}(\mathcal{T} )$. The size of $\mathcal{F}(\mathcal{T} )$ is the number of its components. An agreement forest is said to be the maximum agreement forest (MAF) for $\mathcal{T}$ if it contains the \emph{smallest} number of components among all agreement forests for the given $\mathcal{T}$.


Note that $\mathcal{F}( \mathcal{T} )$ contains the common components of trees in $\mathcal{T}$. Intuitively,  each tree $T$ in $\mathcal{T} $ can be obtained by connecting these components. When connecting components in an agreement forest to form a tree $T \in \mathcal{T}$, sometime there are topological constraints imposed on two components in the agreement forest \cite{BGMS05}.  The following is based on the definitions and observations in \cite{BGMS05}.

We let $G(\mathcal{F}(\mathcal{T}))$ be a graph where the nodes are the components $\{ \mathcal{F}_i\}$ of $\mathcal{F}(\mathcal{T})$, and there is an edge from $F_i$ to $F_j$ if there exists a tree $T\in \mathcal{T}$ within which there is a directed  path from the corresponding node of the root of $F_i$ to the corresponding node of the root of $F_j$.
If $G\left(\mathcal{F}(\mathcal{T})\right)$ is acyclic,  the agreement forest $\mathcal{F}$ is then said to be acyclic. An agreement forest $\mathcal{F}$ is said to the maximum acyclic agreement forest (MAAF) if $\mathcal{F}$ is an acyclic agreement forest with the smallest number of components for the given $\mathcal{T}$. 

One subtle detail about MAF and MAAF is that each tree in $\mathcal{T}$ is assumed to have an outgroup taxon. This outgroup is meant to properly root the trees when we connect the components in an agreement forest. Throughout the paper, we assume such an outgroup is present in $\mathcal{T}$. 
See Fig. \ref{figMAAF} for an example of agreement forests, which is taken from \cite{WUWANG2010}.

\begin{figure}[!htb]
    \centering
    \includegraphics[scale=0.7]{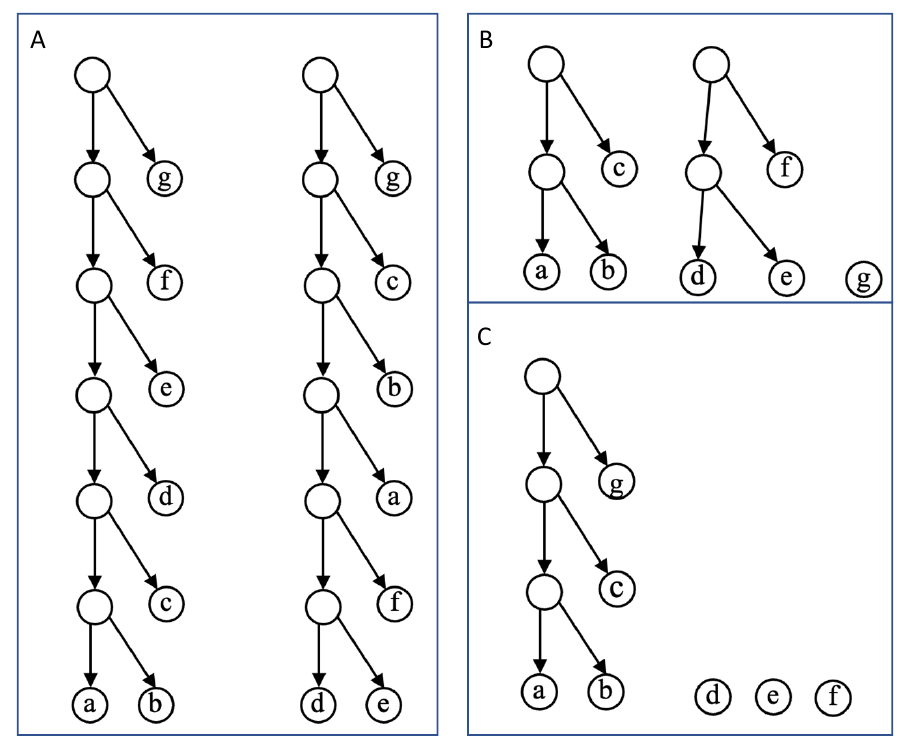}
  \caption{\label{figMAAF} Illustration of the maximum acyclic agreement forest (MAAF) of trees. (A) Two trees.
  $g$ is the added outgroup.  
  (B) A maximum agreement forest (MAF) of the two trees. This MAF is not acyclic, because the first component (with taxa $a,b$ and $c$) is below the second component (with taxa $d,e$ and $f$) in the left tree, but is above the second component in the right tree in Part A, 
  forming a cycle between these two components in the component graph of the agreement forest.
  (C) An MAAF of the two trees.}
\end{figure}


Regarding to the reticulation number of two trees, the following theorem establishes a sharp upper bound.

\begin{theorem}
\label{theorem2Tree}
(\cite{BGMS05}) The reticulation number $R(T,T')$ of two trees $T$ and $T'$ is equal to the size of the MAAF of $T$ and $T'$ minus one, which is at most $ n-2$ if there are $n$ taxa in $T$ and $T'$.  Moreover, the reticulation number for two conjugate caterpillar trees with $n$ taxa is $n-2$. 
\end{theorem}

In the rest of this paper, to simplify notation, we often drop $\mathcal{T}$ and refer to the agreement forest as $\mathcal{F}$, with the understanding that it is always for some given  $\mathcal{T}$.

\subsection{Maximum agreement subtree}

A concept related to the maximum agreement forest (MAF) for a given set of trees $\mathcal{T}$ is the maximum agreement subtree (MAST). Instead of finding a set of components that ``partition'' each tree in $\mathcal{T}$, the MAST problem  aims to find a single component ``contained'' in each $T\in \mathcal{T}$ that is as large as possible.
Specifically, an agreement subtree $T'$ is a spanning subtree over a subset of taxa that appears in each tree $T\in \mathcal{T}$. Notably, each component in an agreement forest is an agreement subtree. 
The size of an agreement subtree is defined to be the number of taxa in it. A MAST has the largest number of taxa.


An interesting question about MAST is: how large can the MAST be for a set of trees with $n$ taxa? For rooted trees, it is easy to see that the MAST for two conjugate caterpillar trees contains only two taxa, making the MAST problem trivial. Therefore, existing research on MAST focuses on unrooted trees. 

\begin{lemma}
\label{lemma2MAT}
(\cite{MARKIN2020612}) The size of the MAST for two unrooted trees with $n$ taxa is $\Theta(  \log(n)  )$.  More precisely, the size of the MAST must be at least $\frac{1}{48}\log(n)$ for large $n$.
\end{lemma}

\subsection{Three or more trees}

Compared to the two-tree case, little is known about the range of the reticulation number of three or more trees. Let $m = |\mathcal{T}|$ and let the trees have $n$ taxa. There is a trivial bound: $R(\mathcal{T}) \le (m-1)(n-2)$ 
\cite{THREETREES23}. We also have the following bound that depends on the reticulation number for a pair of trees in $\mathcal{T}$.
%

\begin{proposition}
\label{propTrivBound}
(\cite{TCLBJCB23})
$\mathcal{T}$ be a set of trees on $n$ taxa and 
$m = \vert \mathcal{T}\vert $. 
If there exists a pair of trees for which the reticulation number is $d$, then, 

\[  R(\mathcal{T}) \le  (m-2)(n-2) + d   \] 
In addition, we can even compute a tree-child network with at most $(m-2)(n-2) + d$ reticulation nodes that displays the trees. 
\end{proposition}

Note that, in the worst case, $d=n-2$ by Theorem~\ref{theorem2Tree}. Thus, Proposition~\ref{propTrivBound} does not seem to lead to better bounds than the 
trivial $(m-1)(n-2)$ bound. However, we will show later that Proposition 2.3 can indeed lead to improved bounds on reticulation numbers with additional facts.


\begin{figure}[!t]
\centering
\includegraphics[scale=0.7]{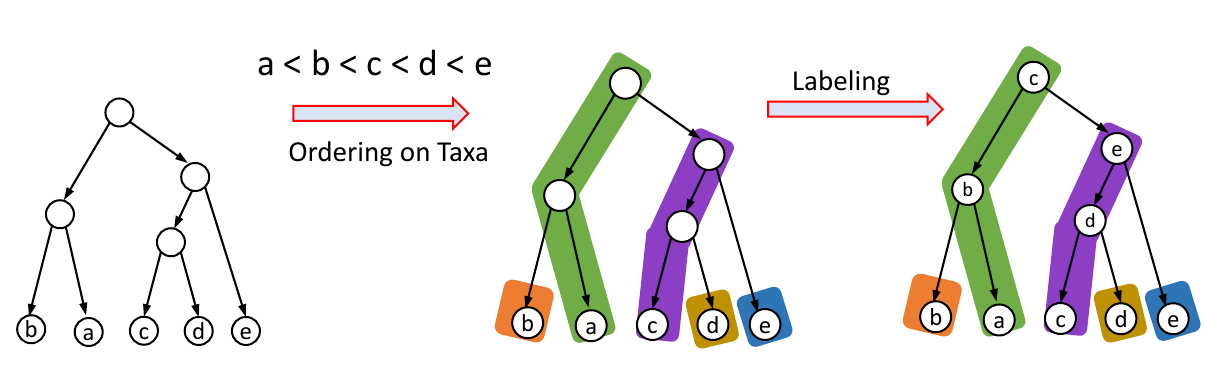}
\caption{\label{fig30}
Illustration of decomposing a tree $T$ into $n$ disjoint paths and the construction of the lineage taxon string for each taxon. Here, taxa are ordered alphabetically. The obtained strings for $a$ to $e$ are respectively $cba$, 
$b$, $edc$, $d$ and $e$.}
\end{figure}

\subsection{Tree-child networks}
\label{sectTreeChild}


Let $T$ be a binary tree on a $n$-taxon set $\mathcal{X}=\{x_1, \ldots, x_n\}$. We now consider an ordering $\prec$ on $\mathcal{X}$ by assuming:
$$x_1\prec x_2 \prec \cdots \prec x_n.$$ 
Using this ordering, we can decompose $T$ into $n$ disjoint paths $P_1, P_2, \cdots, P_n$ through a recursive approach (Fig.~\ref{fig30}): 
\begin{itemize}
    \item $P_1$ is the unique path from the root to the leaf $x_1$. 
   \item For $i>1$, $P_i$ consists of the leaf $x_i$ and  its ancestors that are not in $\cup_{1\leq j\leq i-1}P_j$. 
\end{itemize}
It is clear that each path is non-empty and $P_{n}$ just contains the leaf $x_{n}$.

We further label an internal node $u$ with a taxon $x_j$ if $u$ is the parent of the first node in the path $P_j$ whose last node is $x_i$ (Fig.~\ref{fig30}). This process yields a labeling string from each $P_i$, called the lineage taxon string (LTS) of $x_i$.  Note that the LTS of $x_i$ is a string on 
$\{x_i, x_{i+1}, \cdots, x_{n}\}$.

Suppose we are given $m$ trees $T_j$ ($1\leq j\leq m$) on $\mathcal{X}$ and the ordering $\prec$. 
For each $x_i$, we obtain an LTS $S_{ij}$ in every $T_j$.  A sequence on $\{x_i, x_{i+1}, \cdots, x_{n}\}$ is called a common super-sequence $S$  of these $S_{ij}$ if each $S_{ij}$ can be obtained from $S$ by removing 0 or more symbols.
The following theorem establishes the connection between the lengths of $S_{ij}$ and the reticulation number of a tree-child network that displays these trees $T_j$.

\begin{proposition}
\label{propLTS}
 (\cite{LXZ23})
 Let $\mathcal{T}$ be a set of $m$ trees on $\mathcal{X}$ and let $\prec$ be an ordering on $\mathcal{X}$.
  Suppose, for each $i$, $S_i$ is a common super-sequence of the LTSs $S_{i1}, S_{i2},\cdots, S_{im}$ that are obtained for $x_i$ in the $m$ trees,  and $x_i$ appears only as the last character in $S_i$ and $x_{j}$ does not appear in $S_i$ for each $j$ such that $x_j\prec x_i$. A tree-child network $N$ can be efficiently reconstructed from $S_1, S_2, \cdots, S_n$ such that $N$ displays each $T_i$ and its reticulation number is 
  $\sum_{1\leq i\leq n}(|S_i|-1) - (n-1).$
\end{proposition}

\section{Two useful facts}

We start with two simple but useful observations.

\subsection{A connection between the MAST and the reticulation number}

Intuitively, if two rooted trees have a large rooted agreement subtree, then their maximum acyclic agreement forest cannot contain too many components. Then, by Theorem \ref{theorem2Tree}, the reticulation number for the trees is smaller than  $n-2$. More precisely, we can establish the following result.

\begin{proposition}
\label{propMASMAAF}
Let $T_1$ and $T_2$ be two trees on $n$ taxa.
If there is a rooted agreement subtree of size $k$ for $T_1$ and $T_2$, then,   
\[  R(  T_1, T_2 )  \le n-k  \]

\end{proposition}

\begin{proof}
Let $T_s$ be a rooted agreement \textit{subtree} of $T_1$ and $T_2$ such that it contains $k$ leaves. We construct a forest $\mathcal{F}$ for $T_1$ and $T_2$ which contains $n-k+1$ components:   $T_s$ and $n-k$ components, each with a single leaf out of the $n-k$ taxa that are not in $T_s$. These $n-k$ components are called the singleton components. Here, we assume the added outgroup is part of $T_s$: since the outgroup is present for each tree, this outgroup can clearly be present in any agreement subtree. 
 Since $T_s$ is an agreement subtree, it is obvious that $\mathcal{F}$ is an agreement forest: each of $T_1$ and $T_2$ can be constructed by connecting these $n-k+1$ components. Now we argue that $\mathcal{F}$ is also acyclic. To see this, note that in the graph $G(\mathcal{F})$, there is an edge from the root (the outgroup) of $T_s$ to each singleton component. And there are \emph{no} other edges in $G(\mathcal{F})$: a singleton component cannot be a source of an edge in this graph.  Thus, $G(\mathcal{F})$ is acyclic.
Therefore, the number of components in the maximum acyclic agreement forest is at most $n-k+1$. By Theorem \ref{theorem2Tree}, $R(T_1,T_2) \le (n-k+1)-1 = n-k$.
\end{proof}

Note that existing non-trivial results on MAST are all for unrooted trees, while MAAF assumes rooted trees. In the following section, we will show that Proposition \ref{propMASMAAF} can indeed lead to a slightly improved bound on reticulation number for three trees.

\subsection{More about the reticulation number for two trees}
\label{Sect3}

Let  $T_1$ and $T_2$ be two distinct trees on three taxa $a, b$ and $c$.  Without the loss of generality, we may assume that the leaf $a$ is a child of the root in $T_1$ and $b$ is the child of the root in $T_2$. Hence $T_1$ and $T_2$ are topologically different, as shown in Fig.~\ref{fig1} (left and middle). This implies that $\calR (T_1, T_2)=1$, as the network with one binary reticulation node in Fig.~\ref{fig1} (right) displays $T_1$ and $T_2$, 

Let $T_1$ and $T_2$ be two distinct trees on four taxa $\{a, b, c, d\}$. If the two trees do not share a subtree on three taxa, there are two 
possibilities:  they are either conjugate caterpillar trees (top, Fig.~\ref{fig2}) or  balanced trees (bottom, Fig.~\ref{fig2}). In either case, $\calR (T_1, T_2)=2$, as shown in Fig.~\ref{fig2}.

In general, for trees on five or more taxa, we have the following result.

\begin{figure}[!b]
\centering
\includegraphics[scale=0.6]{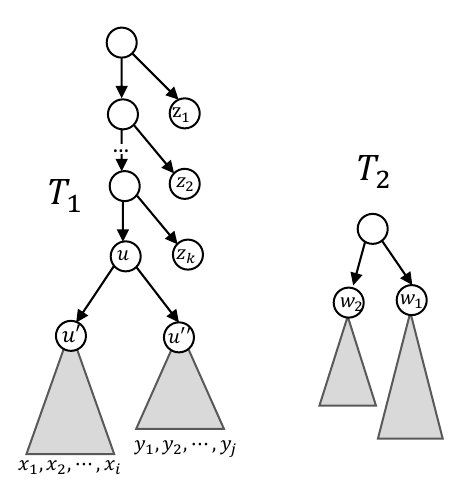}
\caption{\label{fig3}
Illustration of the topological structure of $T_1$ and $T_2$ in the proof of Proposition~\ref{prop_2tree_case}. Here, $T_1$ is a non-caterpillar tree.  
}
\end{figure}

\begin{proposition}
\label{prop_2tree_case}
    Let $T_1$ and $T_2$ be two binary trees on $n$ taxa, where $n\geq 5$.  If $T_1$ is not a caterpillar tree,  $R(T_1, T_2) \leq n-3$. 
\end{proposition}
\begin{proof}
 Let  $T_1$  be  a non-caterpillar tree on $n$ ($\geq 5$) taxa. Then, there exists an internal node $u$ with children $u'$ and $u''$ such that:
 \begin{itemize}
     \item Both $u'$ and $u''$ are internal nodes and so there are at least two leaves below each of $u'$ and $u''$;  and
     \item 
     Each ancestor of $u$ (if it exists) has exactly one leaf child, 
 \end{itemize}
 as shown in Fig.~\ref{fig3} (left).
Assume the subset of leaves below $u'$  is  $\mathcal{X}=\{x_1, x_2, \cdots, x_i\}$ ($i\geq 2$) and the set of leaves below $u''$ is  $\mathcal{Y}=\{y_1, y_2, \cdots, y_j\}$ ($j\geq 2$).  We also assume that the subset of leaves not below $u$ is $\mathcal{Z}=\{z_1, z_2, \cdots, z_k\}$. $T|_{x,y,z}$ denotes the spanning subtree of $T$ which connects the three taxa $x,y$ and $z$. 

Now, we consider $T_2$. Let
the children of the root be $w_1$ and $w_2$ in $T_2$ (right, Fig.~\ref{fig3}).
Using $\mathcal{L}(w_1)$ and $\mathcal{L}(w_2)$ to denote the subset of leaves below $w_1$ and $w_2$ respectively,  we consider the following three cases.

\noindent \textbf{Case 1}: $\mathcal{X}\cap \mathcal{L}(w_1)$ contains at least two leaves $x'$ and $x''$. 

If $\mathcal{L}(w_2)\cap (\mathcal{Y}\cup \mathcal{Z})$ contains at least one leaf $\ell$,
$T_1|_{\{x', x'', \ell\}}$ is a common subtree of $T_1$ and $T_2$. Otherwise,
$(\mathcal{Y}\cup \mathcal{Z}) \subseteq \mathcal{L}(w_1)$. 
In this subcase, 
for $\ell\in \mathcal{L}(w_2)$, 
$T_1|_{\{y_1, y_2, \ell\}}$ is a common subtree of $T_1$ and $T_2$.

\noindent \textbf{Case 2}: $\mathcal{X}\cap \mathcal{L}(w_2)$ contains at least two leaves $x'$ and $x''$. 

This case is symmetric to Case 1. Using a similar argument, we can show that $T_1$ and $T_2$ has a common subtree on 3 taxa.

\noindent \textbf{Case 3}: $\mathcal{X}\cap \mathcal{L}(w_1)$ contains only a leaf $x'$ and $X\cap \mathcal{L}(w_2)$ contains only a leaf $x''$.

If $\mathcal{L}(w_1)\cap \mathcal{Y}$ contains at least two leaves $y', y''$, then
$T_1|_{\{y', y'', x''\}}$ is a common subtree of $T_1$ and $T_2$.
If $\mathcal{L}(w_2)\cap \mathcal{Y}$ contains at least two leaves $y', y''$, then
$T_1|_{\{y', y'', x'\}}$ is a common subtree of $T_1$ and $T_2$. Otherwise, since $n\geq 5$, we conclude that
$\mathcal{L}(w_i) \cap \mathcal{Y}$ contains only one leaf for $i=1, 2$ and $\mathcal{Z}\neq \emptyset$.  Without loss of generality, we assume $\mathcal{L}(w_1)\cap \mathcal{Z}$ contains a leaf $z$. Then,  $z$ and the unique leaf of $\mathcal{X}\cap \mathcal{L}(w_2)$ and the unique leaf of $\mathcal{Y}\cap \mathcal{L}(w_2)$ form a common subtree of $T_1$ and $T_2$.

We have showed that $T_1$ and $T_2$ share a 3-leaf subtree. 
According to Proposition \ref{propMASMAAF}, this implies that $\calR(T_1, T_2)\leq n-3$. 
\end{proof}

Note that Proposition~\ref{prop_2tree_case} does not  hold for $n=4$. Consider $T_1=((a,b),(c,d))$ and $T_2 = ((a,c),(b,d))$ (bottom row, Fig.~\ref{fig2}). Neither of the two trees is a caterpillar tree, yet their reticulation number is $2$, which is greater than $4-3$.

\section{The reticulation number for three trees}
\label{sect4}

We are given three trees $\mathcal{T} = \{ T_1, T_2, T_3\}$, where each $T_i$ are over the same $n$ taxa. We want to give a non-trivial upper bound on the hybridization number $R(T_1, T_2, T_3)$ for the three trees. 
For three trees, the trivial upper bound on the reticulation number \cite{THREETREES23} is $2(n-2)$.
%
Note that there are three caterpillar trees on $n$ taxa for which the reticulation number is at least $3n/2$ \cite{THREETREES23}. In the remainder of this section, we aim to derive an upper bound between between $3n/2$ and $2(n-2)$ for three trees using the following approach:
\begin{quote}
 Proposition 2.3 provides an upper bound for a set of trees based on the reticulation number of a pair of given trees. It does not explicitly imply an improved bound on the reticulation number, as the reticulation number for any two conjugate caterpillar trees on $n$ taxa is $n-2$. However, by applying the result on the MAST for unrooted trees in Lemma 2.2, we can show that for any set of three rooted trees on $n$ taxa, there always exists a pair of trees whose reticulation number is $n-\Theta(\log\log n)$ for large $n$, leading to an improved bound on the reticulation number for the three tree case.
 \end{quote}


Recall that the MAST for two unrooted trees $T$ and $T'$ is a common spanning subtree of $T$ and $T'$ on a subset of taxa. The concept of MAST can be easily extended to three or more trees: the MAST of multiple trees is a spanning subtree on a fixed subset of taxa that is shared by all trees. 

\begin{lemma}
\label{lemma3MATUnroot}
Any three unrooted trees $T_1, T_2, T_3$ on the same $n$ taxa have an unrooted agreement subtree that has  at least $\Theta(    \log\log n    )$ taxa for large $n$. 
\end{lemma}

\begin{proof}
Let $n$ be sufficiently large.
By Lemma \ref{lemma2MAT}, $T_1$ and $T_2$ have an unrooted agreement subtree $Y$ over a subset of taxa $S$ such that  $\vert S\vert =\Theta(  \log n  )$. 
Now, we let $T_3^S$ be  the spanning subtrees of $T_3$ over $S$. That is,  $T_3^S$ is obtained by removing 
from $T_3$ the taxa not in $S$, along with any edges that are not part of the shortest paths between the taxa in $S$.
Applying Lemma \ref{lemma2MAT}  to $Y$ and $T_3^S$, we obtain an agreement subtree $Z$ with size  $\Theta(  \log   \log n  )  $ for $Y$ and $T_3^S$, as $\vert S\vert =\Theta(  \log n  )$. 
Clearly, $Z$ is a common agreement subtree for $T_1, T_2$ and $T_3$.
\end{proof}


\begin{proposition}
\label{theoremMATPair}
For any three rooted binary phylogenetic trees $T_1, T_2$ and $T_3$ (over the same $n$ taxa),  there always exists a pair of trees that have a rooted agreement subtree of size $\Theta( \log\log n ) $.
\end{proposition}

\begin{proof}
Let $T^S$ be an unrooted agreement subtree of size $\Theta(    \log\log n    )$ for the unrooted versions of  the trees $T_1, T_2$  and $T_3$, obtained by using Lemma \ref{lemma3MATUnroot}. 
$T^S$ may induce different rooted subtrees in $T_1, T_2$ and $T_3$, as these subtrees may have their roots in different branches  of $T^S$.
We argue that, regardless of how the subtrees induced by $T^S$ are rooted in the given trees, there are always two induced subtrees that have a \emph{rooted} agreement subtree of size $\Theta(     \log\log n    )$.

\begin{figure}[!b]
    \centering
    \includegraphics[scale=1.0]{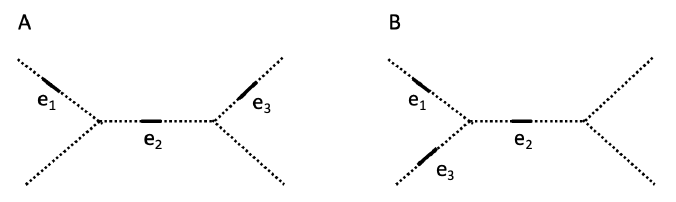}
  \caption{\label{figRootTree} Illustration of the positions of the root of the induced rooted subtrees in $T^S$. Dashed lines: the unrooted subtree $T_s$ of $T_1, T_2$ and $T_3$. $e_1, e_2$ and $e_3$: the edges used to root $T_s$ within $T_1, T_2$ and $T_3$ respectively. (A) $e_2$ is in the shortest path between 
  $e_1$ and $e_3$  in $T^S$.
  (B) None of $e_1, e_2$ or $e_3$ is found in the shortest path between the other two in $T^S$.
  }
\end{figure}

To see this, we assume the induced subtree is rooted on the edge $e_i\in E(T^S)$, $i=1, 2, 3$. 
There are two possible cases.

\noindent \textbf{Case 1}: Either $e_1, e_2$ or $e_3$ is found in the shortest path between the other two in $T^S$. Without loss of generality, assume $e_2$ lies between $e_1$ and $e_3$ (Fig.~\ref{figRootTree}A). Rooting at $e_2$ splits $T^S$ into two parts, with the larger part having a size of $\Theta(    \log\log n   )$. If the right part is larger, 
the subtree induced by rooting at $e_1$ and $e_2$  share the right part, implying the rooted subtrees induced by $T^S$ in $T_1$ and $T_2$ are of size $\Theta( \log\log n  )$. 
Symmetrically, if the left part is larger, we conclude that the induced rooted subtree in $T_2$ and $T_3$ are of size $\Theta( \log\log n  )$. 
\\

\noindent \textbf{Case 2}:  None of $e_1, e_2$ or $e_3$ is found in the shortest path between the other two in $T^S$ (Fig. \ref{figRootTree}B).
We let $v_r$ be the last common node in the shortest paths from $e_1$ to $e_2$ and from $e_1$ to $e_3$. 
Deleting $v_r$ splits $T^S$ into three disjoint parts: $X_1$, $X_2$ and $X_3$, where $X_i$ is the part containing $e_i$.  Clearly, the largest part among $X_1$, $X_2$ and $X_3$ has at least
$\Theta( \log\log n  )$ taxa.  Since  the rooted agreement subtree induced by $T^S$ in any two given trees $T_i$ and $T_j$ share the part $V_k$, where $i, j\in \{1, 2, 3\}, i \not= j$ and $k\in \{1, 2, 3\}\setminus \{i, j\}$,  we conclude 
there exists a pair of trees with a rooted agreement subtree of size $\Theta( \log\log n  )$
%
\end{proof}

We now can prove the following result.


\begin{theorem}
\label{theorem3treebound}
For three rooted binary phylogenetic trees $T_1, T_2$ and $T_3$ on the same $n$ taxa, $R(T_1, T_2, T_3) \le 2n-2 - \Theta(     \log\log(n)   )  $.
\end{theorem}

\begin{proof}
By Proposition \ref{theoremMATPair},  there are two trees (say $T_1$ and $T_2$) with a rooted agreement subtree of size $\Theta(    \log\log n   )  $.  Let the size of this agreement subtree be $w =  \Theta(     \log\log(n)    ) $. 
Then, by Proposition \ref{propMASMAAF}, $R( T_1, T_2  )   \le n -  w    $. 
Therefore, by Proposition \ref{propTrivBound}, $R(T_1,T_2,T_3) \le (n-w) + (n-2) = 2n-2 -  \Theta(    \log\log n     ) $.

\end{proof}

A natural question is: what kind of three trees would lead to the largest reticulation number?
In the two-tree case, the reticulation number of $n-2$ is achieved only for two caterpillar trees according to Proposition~\ref{prop_2tree_case}. One may think the maximum reticulation number for three trees can be obtained when all three trees are caterpillar trees. 
It is unknown whether this is the case. However, we can show that  the reticulation number is much smaller than  $2n-2-\Theta(     \log\log(n) ) $  for three caterpillar trees on the same $n$ taxa.

\begin{proposition}
\label{prop3Caterpillar}
Let $T_1, T_2$ and $T_3$ be three caterpillar trees with $n$ taxa. Then,

\[  R(T_1,T_2,T_3) \le 2n- n^{\frac{1}{3}}-2  \]
\end{proposition}



\begin{proof}
We note that any \emph{unlabeled} caterpillar tree of $n$ leaves has only one non-isomorphic topology. 
Here, each caterpillar tree corresponds to a permutation of the $n$ taxa (with labels from $1$ to $n$). The critical observation is, the maximum agreement subtree (MAST) of any two caterpillar trees correspond to the \emph{longest common subsequence} of the two length-$n$ permutations that correspond to the two trees. It is known that,  for any three permutations on the integers from $1$ to $n$, there are two permutations among the three that have a common subsequence of length at least $n^{\frac{1}{3}}$ (\cite{beame2008value}, Lemma 5.9). So by Proposition \ref{propMASMAAF}, there are two trees out of the given three trees whose reticulation number is at most $n-n^{\frac{1}{3}}$.  Thus, by Proposition \ref{propTrivBound},  $R(T_1,T_2,T_3) \le (n-n^{\frac{1}{3}})+(n-2)  = 2n- n^{\frac{1}{3}}-2 $.
\end{proof}

By Proposition~\ref{prop3Caterpillar}, $R(T_1,T_2,T_3) \le 2n-5$ for three caterpillar trees when $n \ge 27$. 
We now show that this trivial bound can be improved to $2n-5$ for any three (not necessarily caterpillar) trees even for $n\geq 5$. 

\begin{proposition}
\label{prop2nminus5}
For any $n\geq 5$ and any three trees $T_1, T_2$ and $T_3$ on the same $n$ taxa, $R(T_1,T_2,T_3) \le 2n-5$. 
\end{proposition}

\begin{proof} Consider three trees $T_1, T_2, T_3$ on the same set of $n$ taxa. 
If one,  say $T_1$,  is not a caterpillar tree, then $R(T_1, T_2) \le n-3$ (Proposition~\ref{prop_2tree_case}). By Proposition \ref{propTrivBound}, $R(T_1,T_2,T_3) \le (n-2)+(n-3) = 2n-5$. 

Suppose  $T_1, T_2$ and $T_3$ are all caterpillar trees.
 By restricting on a subset of 5 taxa, we can  show that two of the three trees share a common subtree on 3 taxa. To see this,
 we assume that the three caterpillar trees on a set $\mathcal{X}$ of 5 taxa are:
 $C(a_1a_2a_3\{a_4a_5\})$, where $a_1$ is the child of the root and $a_4a_5$ is the unique cherry at the bottom,
 $C(b_1b_2b_3\{b_4b_5\})$, and
 $C(c_1c_2c_3\{c_4c_5\})$.
 Since $|\mathcal{X}|=5$, at least two of $a_1, a_2, b_1, b_2, c_1, c_2$, which from different trees, are identical. Without loss of generality, we may assume that $\{a_1, a_2\}\cap \{b_1, b_2\}$ is non-empty. This implies that 
 $\{a_1, a_2\}\cup \{b_1, b_2\}$ contains at most three different taxa and $\{a_3, a_4, a_5\}\cap \{b_3, b_4, b_5\}$ contains at least two different taxa. Let $x\in \{a_1, a_2\}\cap \{b_1, b_2\}$ and $y, z\in \{a_3, a_4, a_5\}\cap \{b_3, b_4, b_5\}$.
 $(x, (y, z))$ is a common subtree of the first and second trees.
 Therefore, by to Proposition~\ref{propTrivBound}, $R(T_1,T_2,T_3) \le (n-3)+(n-2) = 2n-5$.
 %
\end{proof}


\section{The reticulation number for many trees}

In this section, we consider the general case where  a set of trees on a set of $n$ taxa $\mathcal{X}$,  $\mathcal{T} = \{ T_1, T_2, \ldots, T_m \}$,  is given. 
We  provide a non-trivial upper bound,  derived from an algorithm that constructs a tree-child network that displays $\mathcal{T}$ with a reticulation number smaller than  $(m-1)(n-2)$ for large $m$. We will first describe the algorithm and then analyze the reticulation number of the networks it constructs.


\subsection{Algorithm}
Our algorithm for constructing the tree-child network that displays $\cal T$ consists of two phases.
In the first phase, we consider the spanning subtrees $T_i|_{[1,  t]}$ of the trees in $\mathcal{T}$  over the first $t$ taxa with $t$ to be determined later. We construct a tree-child network $N'$ with a reticulation number of $O(t^3)$ that displays these spanning subtrees. This approach for building the `small' tree-child network is closely related to the method used in \cite{LXZ23}.


In the second phase,  we extend the network $N'$ into a larger network $N$ that displays the trees in $\cal T$. Since each tree $T_i\in {\cal T}$ can be derived from  $T_i|_{[1,  t]}$  by successively inserting the taxa $t+1, t+2, \cdots, n$,  we add a reticulation node with in-degree no greater than $m$ for each of the remaining $n-t$ taxa. More precisely,  Step 9 of the algorithm is implemented as the follows:
\begin{quote}
  Let $N'_0=N'$ and $N'_{k}$ be the tree-child network obtained after adding the 
  taxa $t+1, t+2, \cdots, t+k$ for $k>0$. For each $j$, $N'_k$ displays the spanning subtree $T_j\vert_{[1, t+k]}$.
  Therefore, each edge $e$ of $T_j\vert_{[1, t+k]}$ is mapped to a path $P_e$ consisting of at least one tree edge and some reticulation edges in $N'_{k}$, where $k\geq 0$. 

  Assume $T_j\vert_{[1, t+k+1]}$ is obtained from  $T_j\vert_{[1, t+k]}$ by attaching the leaf $k+1$ onto the edge $e_j$. To construct $N'_{k+1}$ from $N'_{k}$, we attach a reticulation edge onto a chosen tree edge in the corresponding path $P_{e_j}$ for each $j$, as shown in Fig.~\ref{fig:4trees.net}.
\end{quote}




\begin{algorithm}[H]
\caption{(for constructing a network displaying multiple trees in two stages)
}
\label{algoNet2}
\algsetup{indent=2em}
\begin{algorithmic}[1]
\STATE Extract spanning subtrees for the first $t$ taxa from each tree in $\mathcal{T}$, where $t \le m$ is determined by the values of $m$ and $n$ (see the analysis).
\FOR{each taxon $i \in [1,t]$ } 
\STATE Obtain LTS $S_{ij}$ for each taxa $i$ and each tree $T_j$ by ordering the taxa from $1$ to $n$ sequentially.
\STATE Construct a super-sequence $S_i$ of $S_{i1} \ldots S_{im}$. 
\STATE Add $|S_i|$ edges to $N$ so to connect taxon $i$ to the taxa to its right: $i+1, \ldots, t$, following the LTS approach introduced in Section \ref{sectTreeChild}. 
\ENDFOR
\FOR{each taxon $i \in [t+1,n]$} 
\FOR{each tree $T_j$}
\STATE Add into $N$ one edge so that the spanning tree of $T_j$ for the first $i$ taxa is displayed in the network.
\ENDFOR
\ENDFOR
\end{algorithmic}
\end{algorithm}


\begin{figure}[t!]
    \centering
        \centering
        \includegraphics[scale=0.6]{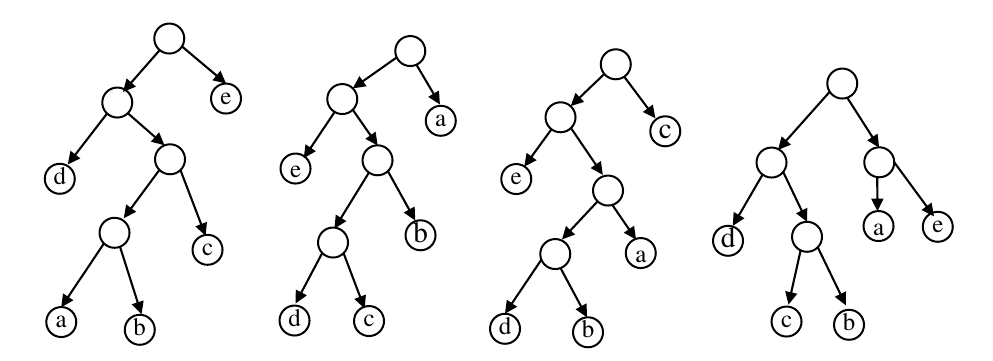}
        \caption{Four trees, with five taxa $a, b, c,d$ and $e$.
      \label{fig:4trees}  }
\end{figure}

\begin{figure}[b!]
        \centering
        \includegraphics[scale=0.7]{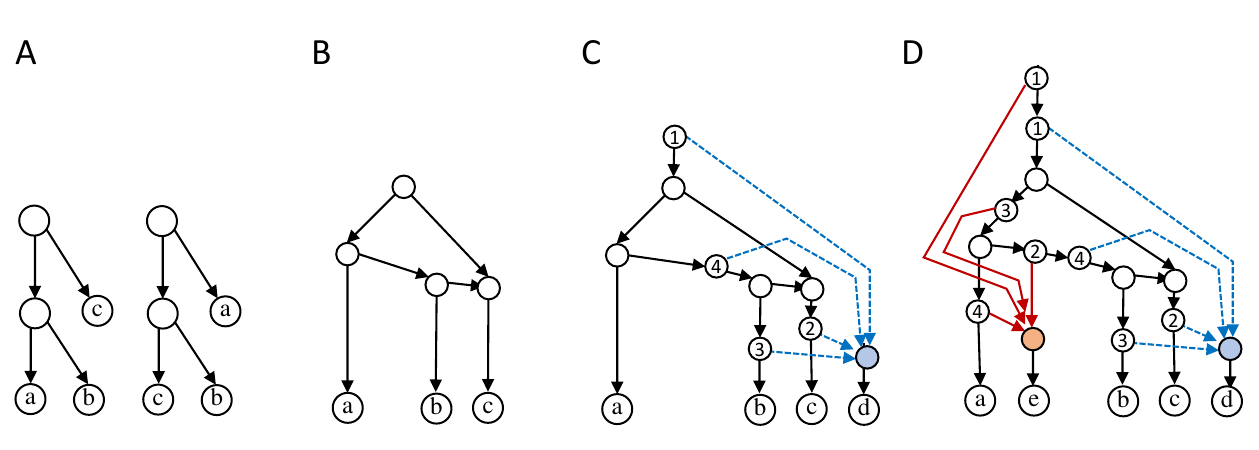}
\caption{ Illustration of construction of a tree-child network for the four trees in Fig. \ref{fig:4trees}. 
(A).  The two subtrees spanned by $\{a, b, c\}$ in the four trees.
(B). A tree-child network that displays the two spanning subtrees, obtained in the first step of  Algorithm 1. 
(C). The network expanded by adding one reticulation node (filled in blue) above the leaf $d$ and one reticulation edge (blue, dashed) per input tree. It contains all the spanning subtrees on $\{a, b,c ,d\}$ of the four trees.
The node labeled with $k$ is the tail of the reticulation edge added to display the spanning subtree of the $k$-th tree (from left to right) for each $k\leq 4$. (D). The resulting network that displays all the trees after the second reticulation node (filled in red) above the leaf $e$ and 4 more reticulation edges (red).
  \label{fig:4trees.net}      }
\end{figure}



We use an example to demonstrate how the algorithm works, where $t$ is set to 3. For the four trees shown in Fig. \ref{fig:4trees}, Algorithm \ref{algoNet2} outputs a network shown in Fig. \ref{fig:4trees.net} that displays these four trees.

There are only two distinct spanning subtrees on $\{a, b, c\}$ in the four trees (Fig.~\ref{fig:4trees.net}A).  In the first phase of the algorithm, a tree-child network (Fig.~\ref{fig:4trees.net}B) over the first $t=3$ taxa $\{a, b, c\}$ is  constructed for the two subtrees (Proposition~\ref{propLTS}).
We use the alphabetic order
$a\prec b\prec c$. The LTSs of $a$ in the left and right subtrees in Fig.~\ref{fig:4trees.net}A are $cba$ and $ba$, 
respectively. Obviously, the shortest common supersequence of $cba$ and $ba$ is $S_a=cba$.  
Since the LTSs of $b$ in the left and right subtrees are $b$ and $cb$, respectively, 
$S_b=cb$ is the shortest common supersequence of these two LTSs. Lastly,
the LTS of $c$ is $c$ for each subtree and thus $S_c=c$ is the shortest common supersequence of the three identical LTSs. $S_a, S_b$ and $S_c$ are the sequences
used in the construction of the tree-child network $N'$ that displays the two spanning subtrees (Fig.~\ref{fig:4trees.net}B).

In the second phase of the algorithm, the remaining two taxa $d$ and $e$ are successively added into  the network $N'$.  For each added taxon,  a reticulation node is introduced and one reticulation edge entering the reticulation node  is added for each tree (Fig.~\ref{fig:4trees.net}C).  The tree edge on which the tail of each added reticulation edge is attached in the case of the taxon being $d$ (resp. $e$)  depends on how to obtain the spanning subtree on $\{a,b, c, d\}$  (resp. $\{a, b, c, d, e\}$) from that on $\{a, b, c\}$ (resp.
$\{a, b, c, d\}$).
Even though the network obtained in this manner may not be the one with the smallest number of reticulation edges, it allows us to derive an improved upper bound using a simple analysis in the next section. 

\subsection{Analysis}

We now estimate the reticulation number of the network constructed by the algorithm.
For taxon $1$ (the first), we  come up with a short string $S_1$ over $\{x_2, \ldots, x_t\}$ that is a common super-sequence of the LTSs of $x_1$ on the spanning trees of each $T_i$ ($1 \le i \le m$). In the worst case, the length of $S_1$ is $(t-1)^2$ (one way to achieve this is simply concatenating $t-1$ copies of string $x_2x_3 \ldots x_t$, as the LTS for $x_1$ in each tree contains $x_i$ at most once for each $2\leq i\leq t$ (Lemma 1, \cite{BZ24}; also see \cite{LXZ23}).
More generally, for each taxon $t'$ ($1 \le t' \le t$), we can construct a common super-sequence $S_{t'}$ of the LTSs of $x_{t'}$ in the spanning subtrees $T_i\vert_{[1, t]}$ such that $\vert S_{t'}\vert \leq (t-t')^2$. 
Therefore, the total number of edges in the constructed tree-child network that displays the spanning subtrees $T_j\vert_{[1, t]}$ is at most: 
$$ \sum_{i=1}^t  (t-i)^2 = \sum_{i=1}^{t-1} i^2 = \frac{  (t-1)t(2t-1)  }{6}.$$

In the second phase,  for each taxon $x_j$ ($t+1 \le j \le n$), we need  to add at most one reticulation edge for each $T_j$.  

In summary, the reticulation number in the network is:
\begin{align}
\label{equPartialNet}
  R(\mathcal{T}) \le  \frac{  (t-1)t(2t-1)  }{6}   + (n-t) m  - (n-1)  = \frac{2t^3 - 3 t^2 +(1-6m)t}{6}  + (m-1)n +1,
\end{align}
where $n-1$ is subtracted because of the minus one term in the definition of the reticulation number in  Eq.~(\ref{defRN}) (Proposition~\ref{propLTS}).
As a function of $t$, the right-hand side of Inequality~(\ref{equPartialNet}) decreases  when $t < t_0 = \frac{1}{2} + \sqrt{  m + \frac{1}{12}  }$ and increases  when $t>t_0$, and  is therefore minimized at $t=t_0$. 
We have the following two cases.


\subsection*{The case $t_0 \le n$}

In this case,  $m \le  \frac{1}{4}  (2n-1)^2 - \frac{1}{12}= n(n-1) +\frac{1}{6} $. This is the case where the number of trees $m$ is not very large. 
In this case, setting
$t= \lceil t_0\rceil \approx \sqrt{m} + 1   $, we obtain:
\begin{eqnarray*} 
R(\mathcal{T}) &\le &    (m-1)n + 1 -   \sqrt{m}  (   \frac{2}{3}m + \frac{1}{2} \sqrt{m} -\frac{1}{6} )\\
&= & 
(m-1)(n-2) - \left[ \sqrt{m} \left(  \frac{2}{3}m - \frac{3}{2} \sqrt{m} - \frac{1}{6}  \right)  + 1\right],     
\end{eqnarray*}
%
implying an $\Theta(m\sqrt{m})$ improvement over the known $(m-1)(n-2)$ bound.

\subsection*{The case $t_0 \geq n$}

In this case,  $m \geq  \lceil n(n-1) + \frac{1}{6}  \rceil =  n(n-1)+1$, that is,  $m  = \Omega( n^2 )$.
Note that the right-hand side of Inequality \ref{equPartialNet} is decreasing when $t\leq t_0$ and $t$ cannot be larger than $n$. So by letting $t=n$, we obtain the best bound:

\begin{eqnarray*} R(\mathcal{T}) &\le&   \frac{  (n-1)n(2n-1)  }{6}     - (n-1) \\
 & = & n(n-1)(n-2) -  \frac{ (n-1)(  4n^2-11n+6  ) }{6}   \\ 
 &\leq & (m-1)(n-2)  - \frac{ (n-1)(  4n^2-11n+6  ) }{6} 
\end{eqnarray*}
implying an improvement of $\Theta(n^3)$ in the reticulation number. 




To summarize, we have the following result.

\begin{theorem}
\label{theoremLargeNumTrees}
For a set $\mathcal{T}$  of $m$ trees  on $n$ taxa, 

\begin{equation*}
    R(  \mathcal{T} )   \le    (m-1)(n-2) -    \begin{cases}
 \Theta(  m \sqrt{m} )   & \text{if }  m \le n(n-1) +\frac{1}{6}, \\
\Theta( n^3)  &\text{otherwise.}
\end{cases}
\end{equation*}

\end{theorem}

\section{Concluding discussions}

In this paper, we present several upper bounds on the reticulation number required to display multiple trees in a phylogenetic network. Our results improve upon the previously known upper bound of $(m-1)(n-2)$  in a non-trivial way for three or more trees on $n$ taxa. In \cite{THREETREES23}, the authors inquired about the comparison between the so-called multi-labeled (MUL) tree model and the phylogenetic network model for tree display.  The question was whether the phylogenetic network model could offer a more efficient representation by using significantly fewer reticulations than the trivial upper bound of $(m-1)(n-2)$.  Our results indicate that this trivial upper bound can indeed be improved using the recent results and techniques developed in the study of tree-child networks,  demonstrating that reticulations can be more effective than multiple labels in representing multiple trees. Several open problems related to reticulation numbers for multiple trees remain for future investigation.

First, we note that the bound in Theorem \ref{theorem3treebound} is only asymptotic for larger $n$. 
It  does not give useful bounds for small $n$. Moreover, the upper bound $2n-n^{1/3}-2$ is proved for three caterpillar trees on $n$ taxa. 
Can the upper bound of  $2n-\Omega(\log n  )$ or  even $2n-\Omega(n^d )$ ($d\leq 1$)  be proved for any three  trees on $n$ taxa?

Another interesting  research direction is the study of special cases involving three trees. Proposition 4.4 demonstrates that a significantly sharper bound can be established for three caterpillar trees compared to the trivial bound. However, our small-scale empirical study using the PIRNs program  \cite{WUMIRRN16} suggests that three caterpillar trees do not yield the maximum reticulation number. Instead, it appears that two caterpillar trees combined with one fully symmetric tree often result in a higher reticulation number than three caterpillar trees alone. Therefore, we conjecture that the maximum reticulation number for $2^k$ taxa,  is achieved with two caterpillar trees plus the fully balanced tree in which the leaves are equally distant from the root.


Can the  bound given in Theorem \ref{theoremLargeNumTrees}  be further improved significantly?  One plausible approach is improving the constructive proof of Theorem \ref{theoremLargeNumTrees}:  finding a way to construct some phylogenetic network with a special structure that uses a smaller number of edges than the one in the proof of Theorem \ref{theoremLargeNumTrees}.

 It is easy to see that two conjugate caterpillar trees on $n$ taxa do not share a common subtree with 3 taxa. However, in Sections~\ref{Sect3} and \ref{sect4}, we have observed for any three trees on $n$ taxa,  there always exist two trees out of these trees that share a common subtree on $3$ taxa (see Proposition~\ref{prop2nminus5}). This raises the following Ramsey-type problem for phylogenetic trees:
\begin{quote}
    Given integers $k$ and $n$ such that $n>k$,
     what is the minimum value of $M(k, n)$ such that for any $M(k, n)$ phylogenetic trees on $n$ taxa, there always exist two of them sharing a common rooted subtree on a subset of $k$ taxa?
\end{quote}
Note that there are $m_k=\frac{(2k-2)!}{2^{k-1}(k-1)!}$ trees on a set of $k$ taxa. Thus, for any $m>m_k$ phylogenetic trees on $n$ taxa, there exist two of them that share a common subtree on $X$ for any set $X$ of $k$ taxa by the Pigeonhole principle.  Proposition~\ref{prop_2tree_case} and Proposition~\ref{prop2nminus5} imply that $M(3, n)=3$ for $n \ge 5$. 



\subsection*{Acknowledgments} 

Research is partly supported by U.S. NSF grant IIS-1909425 (YW) and Singapore Ministery of Education {A-8001951-00-00 } (LZ). The work was started while YW was  attending the Mathematics of Evolution program held at the Institute for Mathematical Sciences,  National University of Singapore,  in September 2023.

\subsection*{Declarations}
The authors have no competing interests to declare that are relevant to the content of this article.

\subsection*{Data availability statements}
This article has no associated data.

%
\bibliographystyle{elsarticle-harv}

\end{document}